\documentclass[12pt]{article}

\usepackage{geometry}
\usepackage{alg}
\usepackage{cite}
\usepackage{booktabs}
\usepackage{paralist}
\usepackage{graphicx}
\usepackage{subfigure}
\usepackage{amssymb}
\usepackage{amsmath}
\allowdisplaybreaks[1]

\usepackage{amsthm}

\newtheorem{theorem}{Theorem}
\newtheorem{lemma}[theorem]{Lemma}
\newtheorem{corollary}[theorem]{Corollary}

\newtheorem{observation}[theorem]{Observation}

\makeatletter
\def\@endtheorem{\endtrivlist}
\makeatother

\newcounter{vcrule}
\newenvironment{vcrule}{\refstepcounter{vcrule}\par\smallskip\noindent
\textbf{VC.\arabic{vcrule}}\quad}{} 

\newcounter{rrule}
\newenvironment{rrule}{\refstepcounter{rrule}\par\smallskip\noindent
\textbf{R\arabic{rrule}}\quad}{} 

\newcounter{brule}
\newenvironment{brule}{\refstepcounter{brule}\par\smallskip\noindent
\textbf{B\arabic{brule}}\quad}{} 

\newcommand{\cvd}{\textsc{Cluster Vertex Deletion}}

\newcommand{\hs}{\textsc{3-Hitting Set}}

\newcommand{\Hv}{H_v}
\newcommand{\Hw}{H_w^{G-v}}
\newcommand{\Fa}{\mathcal{F}}
\newcommand{\F}[2]{\mathcal{F}_{#1,#2}}
\newcommand{\Fv}{\F{v}{k}}

\begin{document}

\title{Faster parameterized algorithm for Cluster Vertex Deletion}
\author{Dekel Tsur%
\thanks{Ben-Gurion University of the Negev.
Email: \texttt{dekelts@cs.bgu.ac.il}}}
\date{}
\maketitle

\begin{abstract}
In the \cvd\ problem the input is a graph $G$ and an integer $k$.
The goal is to decide whether there is a set of vertices $S$ of size at most
$k$ such that the deletion of the vertices of $S$ from $G$ results a graph
in which every connected component is a clique.
We give an algorithm for \cvd\ whose running time is $O^*(1.811^k)$.
\end{abstract}


\section{Introduction}
A graph $G$ is called a \emph{cluster graph} if every connected component of $G$
is a clique (i.e., a complete graph).
A set of vertices $S$ in a graph $G$ is called a \emph{cluster deletion set}
of $G$ if deleting the vertices of $S$ from $G$ results a cluster graph.
In the \cvd\ problem the input is a graph $G$ and an integer $k$.
The goal is to decide whether there is a cluster deletion set of size
at most $k$.

Note that a graph $G$ is a cluster graph if and only if $G$ does not contain
an induced path of size~3.
As \cvd\ is equivalent to the problem of finding whether there is a set of
vertices of size at most $k$ that hits every induced path of size~3 in $G$,
the problem can be solved in $O^*(3^k)$-time~\cite{cai1996fixed}.
A faster $O^*(2.26^k)$-time algorithm for this problem was given by
Gramm et al.~\cite{gramm2004automated}.
The next improvement on the parameterized complexity of the problem came from
results on the more general \hs\ problem~\cite{fernau2010top,wahlstrom2007algorithms}.
The currently fastest parameterized algorithm for \hs\ runs in $O^*(2.076^k)$
time~\cite{wahlstrom2007algorithms}, and therefore \cvd\ can be solved within
this time.
Later, H{\"u}ffner et al.~\cite{huffner2010fixed} gave an $O^*(2^k)$-time
algorithm for \cvd\ based on iterative compression.
Finally, Boral et al.~\cite{boral2016fast} gave an $O^*(1.911^k)$-time
algorithm.

In a recent paper, Fomin et al.~\cite{fomin2016exact} showed a general approach
for transforming a parameterized algorithm to an exponential-time algorithm
for the non-parameterized problem.
Using this method on the algorithm of Boral et al.\ gives an $O(1.477^n)$-time
algorithm for \cvd. This improves over the previously fastest exponential-time
algorithm for this problem~\cite{fomin2010iterative}.
A related problem to \cvd\ is the \textsc{3-Path Vertex Cover} problem.
In this problem, the goal is to decide whether there is a set of
vertices of size at most $k$ that hits every path of size~3 in $G$.
Algorithms for \textsc{3-Path Vertex Cover}  were given
in~\cite{tu2015fixed,wu2015measure,katrenivc2016faster,chang2016fixed,xiao2017kernelization,tsur2018parameterized}.
The currently fastest algorithm for this problem has $O^*(1.713^k)$
running time~\cite{tsur2018parameterized}.

In this paper we give an algorithm for \cvd\ whose running time is
$O^*(1.811^k)$.
Using our algorithm with the method of~\cite{fomin2016exact} gives an
$O(1.448^n)$-time algorithm for \cvd.

Our algorithm is based on the algorithm for Boral et al.~\cite{boral2016fast}.
The algorithm of Boral et al.\ works as follows.
The algorithm chooses a vertex $v$ and then constructs a family of sets,
where each set hits all the induced paths in $G$ that contain $v$.
Then, the algorithm branches on the constructed sets.
In order to analyze the algorithm, Boral et al.\ used a Python script for
automated analysis of the possible cases that can occur in the subgraph of $G$
induced by the vertices with distance at most 2 from $v$.
For each case, the script generates a branching vector and computes the
branching number.
Our improvement is achieved by first making several simple but crucial
modifications to the algorithm of Boral et al.
Then, we modify the Python script by adding restrictions on the cases the
algorithm can generate.
Finally, we manually examine the four hardest cases and for each case we
either show that the case cannot occur, or give a better branching vector for
the case.

\section{Preleminaries}

Let $G = (V,E)$ be a graph.
For a vertex $v$ in a $G$,
$N_G(v)$ is the set of neighbors of $v$,
$\deg_G(v) = |N_G(v)|$, and
$N_G^2(v)$ is the set of all vertices with distance exactly $2$ from $v$.
For a set of vertices $S$, $G[S]$ is the subgraph
of $G$ induced by $S$ (namely, $G[S]=(S,E\cap (S\times S))$).
We also define $G-S = G[V\setminus S]$.
For a set that consists of a single vertex $v$, we write $G-v$ instead of
$G-\{v\}$.

An \emph{$s$-star} is a graph with vertices $v,v_1,\ldots,v_s$ and
edges $(v,v_1),\ldots,(v,v_s)$.
The vertex $v$ is called the \emph{center} of the star,
and the vertices $v_1,\ldots,v_s$ are called \emph{leaves}.

A \emph{vertex cover} of a graph $G$ is a set of vertices $X$ such that
every edge of $G$ is incident on at least one vertex of $X$.

\section{The graph $\Hv$}
Let $v$ be a vertex of $G$.
We define a graph $\Hv^G$ as follows.
The vertices of $\Hv^G$ are $N_1 \cup N_2$, where $N_1 = N_G(v)$ and
$N_2 = N_G^2(v)$.
For $u\in N_1$ and $u'\in N_2$, there is an edge $(u,u')$ in $\Hv^G$
if and only if $(u,u')$ is an edge in $G$.
Additionally, for $u,u' \in N_1$, there is an edge $(u,u')$ in $\Hv^G$
if and only if $(u,u')$ is \emph{not} an edge in $G$.
Note that $N_2$ is an independent set in $\Hv^G$.
We will omit the superscript $G$ when the graph $G$ is clear from the context.

Two vertex covers $X,X'$ of $\Hv$ are \emph{equivalent} if
$|X| = |X'|$ and $X\cap N_2 = X'\cap N_2$.
We say that a vertex cover $X$ of $\Hv$ \emph{dominates} a vertex cover $X'$
if $|X| \leq |X'|$, $X\cap N_2 \supseteq X'\cap N_2$, and $X,X'$ are not
equivalent.
An equivalence class $\mathcal{C}$ of vertex covers is called \emph{dominating}
if for every vertex cover $X \in \mathcal{C}$,
there is no vertex cover that dominates $X$, and
there is no proper subset of $X$ which is a vertex cover.
and for every nonempty $Z \subseteq X \cap N_2$,
there is no vertex cover that is equivalent to $X\setminus Z$.
A \emph{dominating family} of $v$ is a family $\Fa$ of vertex covers
of $\Hv$ such that $\Fa$ contains a vertex cover from each dominating
equivalence class.
The algorithm of Boral et al.\ is based on the following Lemma.
\begin{lemma}[Boral et al.~\cite{boral2016fast}]\label{lem:dominating-family}
Let $\Fa$ be a dominating family of $v$.
There is a cluster deletion set $S$ of $G$ of minimum size such that
either $v\in S$ or there is $X \in \Fa$ such that $X \subseteq S$.
\end{lemma}

A family $\Fa$ of vertex covers of $\Hv$ is called \emph{$k$-dominating} if
there is a dominating family $\Fa'$ such that
$\Fa = \{X\in \Fa' \colon |X| \leq k \}$.
From Lemma~\ref{lem:dominating-family} we obtain the following simple branching
algorithm for \cvd.
Given an instance $(G,k)$, choose a vertex $v$ and compute a $k$-dominating
family $\Fv$ of $v$.
Then, recursively run the algorithm on the instance $(G-v,k-1)$
(corresponding to a cluster deletion set that contains $v$)
and on the instances $(G-X,k-|X|)$ for every $X \in \Fv$.
In Section~\ref{sec:cvd-alg} we will give a more complex algorithm based on
this idea.

A connected component $C$ of $\Hv$ is called a \emph{seagull} if
$\Hv[C]$ is a 2-star whose center is in $N_1$ and its leaves are in $N_2$.
A subgraph $H$ of $\Hv$ is called an \emph{$s$-skien} if $H$ contains
$s$ seagulls, and the remaining connected components of $H$ are
isolated vertices.
If $\Hv$ is an $s$-skien we also say that $\Hv$ is a skien.

\section{Algorithm for finding a dominating family}\label{sec:dominating-alg}
\newcommand{\vcalgname}{\mathrm{VCalg}}
\newcommand{\vcalg}[2]{\vcalgname(#1,#2)}

This section describes an algorithm, denoted $\vcalgname$, for constructing a
$k$-dominating family $\Fv$ of $\Hv$.
The algorithm is a recursive branching algorithm,
and it is based on the algorithm from~\cite{boral2016fast}.
Let $(H,k)$ denote the input to the algorithm.
When we say that the algorithm \emph{recurses on} $X_1,\ldots,X_t$,
the algorithm performs the following lines.
\begin{algtab}
$L \gets \emptyset$.\\
\algforeach{$i = 1,\ldots,t$}
 \algforeach{$X \in \vcalg{H-X_i}{k-|X_i|}$}
  Add $X\cup S_i$ to $L$.\\
 \algend
\algend
\algreturn $L$
\end{algtab}

Given an input $(H,k)$, the algorithm applies the first applicable rule from
the rules below.

\begin{vcrule}
If $k < 0$, return an empty list.
\label{vcrule:terminate-1}
\end{vcrule}

\begin{vcrule}
If $H$ does not have edges, return a list with a single element which is
an empty set.
\label{vcrule:terminate-2}
\end{vcrule}

\begin{vcrule}
If there is a vertex $u \in N_1$ such that $\deg_H(u) = 1$,
recurse on $\{w\}$, where $w$ be the unique neighbor of $u$.
\label{vcrule:degree-1}
\end{vcrule}

\begin{vcrule}
If $C={u_1,\ldots,u_s}$ is a cycle in $H$ such that 
$\deg_H(u_i) = 2$ for all $i$ and $C \subseteq N_1$,
recurse on $\{u_i \colon i\text{ is odd}\}$.
\label{vcrule:N1-cycle}
\end{vcrule}

\begin{vcrule}
If $C={u_1,\ldots,u_{2s}}$ is an even cycle in $H$ such that
$\deg_H(u_i) = 2$ for all $i$, $u_i \in N_1$ for odd $i$, and
$u_i \in N_2$ for even $i$, recurse on $C\cap N_2$.
\label{vcrule:alternating-cycle}
\end{vcrule}

\begin{vcrule}
If $H$ contains vertices of degree at least 3, choose a vertex $u$ as follows.
Let $d$ be the maximum degree of a vertex in $\Hv$.
If there is a vertex with degree $d$ in $N_2$, let $u$ be such vertex.
Otherwise, $u$ is a vertex with degree $d$ in $N_1$.
Branch on $\{u\}$ and on $N_H(u)$.
\label{vcrule:degree-3}
\end{vcrule}

Note that if Rules~VC.1--VC.\ref{vcrule:degree-3}
cannot be applied, every vertex in $N_1$ has degree 2 and every vertex in $N_2$
has degree~1 or~2.
Additionally, every connected component in $H[N_1]$ is an induced path.

\begin{vcrule}
If $N_1$ is not an independent set, let $C$ be a connected component of $H[N_1]$
with minimum size among the connected components of size at least~2.
$H[N_1]$ is a path $u_1,\ldots,u_s$,
and let $u_0$ and $u_{s+1}$ be the unique neighbors of
$u_1$ and $u_s$ in $N_2$, respectively.
Branch on $\{u_i \colon i\text{ is even}\}$ and
$\{u_i \colon i\text{ is odd}\}$.
\label{vcrule:N1-path}
\end{vcrule}

We note that the reason we choose a connected component with minimum size
is to simplify the analysis. The algorithm does not depend on this choice.

If Rules~VC.1--VC.\ref{vcrule:N1-path} cannot be applied,
every connected component of $H$ is an induced path $u_1,\ldots,u_{2s+1}$
such that $u_i \in N_2$ if $i$ is odd and $u_i \in N_1$ if $i$ is even.
\begin{vcrule}
Otherwise, let $C$ be a connected component of $H$ with maximum size.
Branch on $C\cap N_1$ and $C\cap N_2$.
\label{vcrule:alternating-path}
\end{vcrule}

Note that the branching vectors of the branching rules of the algorithm
are at least $(1,2)$. The branching vector $(1,2)$ occurs only when $H$
is a skein.
In this case, the algorithm applies Rule~VC.\ref{vcrule:alternating-path} on
some seagull.

The differences between the algorithm in this section and the algorithm
in~\cite{boral2016fast} are as follows.
\begin{enumerate}
\item
In Rule~VC.\ref{vcrule:degree-3}, the algorithm of~\cite{boral2016fast} chooses
an arbitrary vertex $u$ with degree at least~3.
\item
Rule~VC.\ref{vcrule:N1-path} is different than the corresponding rule
in~\cite{boral2016fast}.
\item
Rule VC.\ref{vcrule:N1-cycle} does not appear in~\cite{boral2016fast}.
\end{enumerate}

\section{The main algorithm}\label{sec:cvd-alg} 
\newcommand{\cvdalgname}{\mathrm{CVDalg}}
\newcommand{\cvdalg}[2]{\cvdalgname(#1,#2)}
\newcommand{\twins}[1]{\mathrm{Twins}(#1)}
\newcommand{\vc}[1]{\mathrm{vc}(#1)}

In this section we describe the algorithm for \cvd.

We say that vertices $v$ and $v'$ are \emph{twins} if $N[v] = N[v']$.
Note that $v'$ is a twin of $v$ if and only if $v'$ is an isolated vertex
in $\Hv$.
Let $\twins{v}$ be a set containing $v$ and all its twins.
\begin{lemma}\label{lem:twins}
If $v,v'$ are twins then for every cluster deletion set $S$ of $G$ of
minimum size,
$v \in S$ if and only if $v' \in S$.
\end{lemma}
\begin{proof}
Suppose conversely that there is a cluster deletion set $S$ of $G$ of
minimum size such that, without loss of generality,
$v\in S$ and $v' \notin S$.
Let $S' = S \setminus \{v\}$. 
We claim that $S'$ is a cluster deletion set.
Suppose conversely that $S'$ is not a cluster deletion set.
Therefore, there is an induced path $P$ of size~3 in $G-S'$.
Since $S$ is a cluster deletion set, $P$ must contain $v$.
Since $v$ and $v'$ are twins, $P$ does not contain $v'$.
Therefore, replacing $v$ with $v'$ gives an induced path $P'$, and
$P'$ is also an induced path in $G-S$, a contradiction to the assumption
that $S$ is a cluster deletion set.
Therefore, $S'$ is a cluster deletion set.
This is a contradiction to the assumption that $S$ is a cluster deletion set
of minimum size.
Therefore, the lemma is correct.
\end{proof}
The following lemma generalizes Lemma~9 in~\cite{boral2016fast}.
\begin{lemma}
Let $X$ be a vertex cover of $\Hv$.
There is a cluster deletion set of $G$ of minimum size such that either
$v\notin S$ or $|X \setminus S| \geq 1+|\twins{v}|$.
\end{lemma}
\begin{proof}
Let $S$ be a cluster deletion set of $G$ of minimum size and suppose that
$v\notin S$ and $|X \setminus S| \leq |\twins{v}|$ otherwise we are done.
By Lemma~\ref{lem:twins}, $S \cap \twins{v} = \emptyset$.
Let $S' = (S \setminus \twins{v}) \cup X$.
We claim that $S'$ is a cluster deletion set of $G$.
Suppose conversely that $S'$ is not a cluster deletion set.
Then, $G-S'$ contains an induced path $P$ of size~3.
$P$ contains exactly one vertex $v' \in \twins{v}$.
Since $X \subseteq S'$ and $X$ is a vertex cover of $\Hv$, we have that
the connected component of $v$ in $G-S'$ is a clique
(by Lemma~6 in~\cite{boral2016fast}) and this component contains $v'$.
This is a contradiction, so $S'$ is a cluster deletion set of $G$.
From the assumption $|X \setminus S| \leq |\twins{v}|$ we obtain that
$S'$ is a cluster deletion set of $G$ of minimum size.
Since $v\notin S'$, the lemma is proved.
\end{proof}
Denote by $\vc{\Hv}$ the minimum size of a vertex cover of $\Hv$
\begin{corollary}
If $|\twins{v}| \geq \vc{\Hv}$ then
there is a cluster deletion set of $G$ of minimum size such that $v\notin S$.
\end{corollary}

The algorithm for \cvd\ is a branching algorithm.
Let $(G,k)$ denote the input to the algorithm.
We say that the algorithm \emph{branches on} $S_1,\ldots,S_t$ if for each $S_i$,
the algorithm tries to find a cluster deletion set $S$ that contains $S_i$.
More precisely, the algorithm performs the following lines.
\begin{algtab}
\algforeach{$i = 1,\ldots,t$}
 \algifthen{$\cvdalg{G-S_i}{k-|S_i|}$ returns `yes'}{\algreturn `yes'}
\algend
\algreturn `no'.
\end{algtab}

Given an instance $(G,k)$ for \cvd, the algorithm first repeatedly applies
the following reduction rules.

\begin{rrule}
If $k < 0$, return `no'.
\end{rrule}

\begin{rrule}
If $G$ is a cluster graph, return `yes'.
\end{rrule}

\begin{rrule}
If there is a connected component $C$ which is a clique,
delete the vertices of $C$.
\label{rrule:clique}
\end{rrule}

\begin{rrule}
If there is a connected component $C$ such that there is a vertex $v\in C$
for which $G[C]-v$ is a cluster graph,
delete the vertices of $C$ and decrease $k$ by 1.
\label{rrule:size-1}
\end{rrule}

\begin{rrule}
If there is a connected component $C$ such that the maximum degree of $G[C]$
is 2, compute a cluster deletion set $S$ of $G[C]$ of minimum size.
Delete the vertices of $C$ and decrease $k$ by $|S|$.
\label{rrule:path-or-cycle}
\end{rrule}

When the reduction rules cannot be applied, the algorithm chooses a vertex
$v$ as follows.
If the graph has vertices with degree 1, $v$ is a vertex with degree~1.
Otherwise, $v$ is a vertex with maximum degree in $G$.
The algorithm then constructs the graph $\Hv$ and computes a $k$-dominating
family $\Fv$ of $\Hv$ using algorithm $\vcalgname$ of
Section~\ref{sec:dominating-alg}.
Additionally, the algorithm decides whether $\vc{\Hv}$ is 1, 2, or at least 3 
(this can be done in $n^{O(1)}$ time).
It then performs one of the following branching rules, depending on $\vc{\Hv}$.

\begin{brule}
If $\vc{\Hv} = 1$ or ($\vc{\Hv} = 2$ and $|\twins{v}|\geq 2$),
branch on every set in $\Fv$.
\label{brule:vc1}
\end{brule}

\begin{brule}
If $\vc{\Hv} = 2$ and $|\twins{v}| = 1$,
let $X$ be a vertex cover of size~2 of $\Hv$,
and let $w \in X$ be a vertex such that the connected component of $w$ in $G-v$
is not a clique.
Construct a $(k-1)$-dominating family $\F{w}{k-1}$ of the
graph $\Hw$ using algorithm $\vcalgname$.
Branch on every set in $\Fv$ and on $\{v\} \cup S$ for every $S \in \F{w}{k-1}$.
\label{brule:vc2}
\end{brule}

\begin{brule}
If $\vc{\Hv} \geq 3$ branch on $\twins{v}$ and on every set in $\Fv$.
\label{brule:vc3}
\end{brule}

Note that if $v$ has degree~1, $\vc{\Hv} = 1$ and therefore
Rule~B\ref{brule:vc1} is applied.

The main difference between the algorithm in this section and the algorithm
in~\cite{boral2016fast} are as follows.
\begin{enumerate}
\item
The algorithm of~\cite{boral2016fast} chooses an arbitrary vertex $v$.
\item
The algorithm of~\cite{boral2016fast} does not take advantage of twins.
That is, in Rule~B\ref{brule:vc3}, the algorithm of~\cite{boral2016fast}
branches on $\{v\}$ instead of $\twins{v}$.
\item
Rule~R\ref{rrule:path-or-cycle} does not appear in~\cite{boral2016fast}.
\end{enumerate}

\section{Analysis}
\newcommand{\tree}[3]{T_{#1}(#2,#3)}
\newcommand{\bv}[1]{\mathrm{bv}(#1)}
\newcommand{\vctree}[2]{\tree{\vcalgname}{#1}{#2}}
\newcommand{\vctreeb}[3]{T_{#1}(#2,#3)}

In this section we analyze our algorithm.

Let $A$ be some parameterize algorithm on graphs.
The run of the algorithm on an input $(G,k)$ can be represented
by a \emph{recursion tree}, denoted $\tree{A}{G}{k}$, as follows.
The root $r$ of the tree corresponds to the call $A(G,k)$.
If the algorithm terminates in this call, the root $r$ is a leaf.
Otherwise, suppose that the algorithm is called recursively on the instances
$(G_1,k-a_1),\ldots,(G_t,k-a_t)$.
In this case, the root $r$ has $t$ children. The $i$-th child of $r$
is the root of the tree $\tree{A}{G_i}{k-a_i}$.
The edge between $r$ and its $i$-th child is labeled by $a_i$.
See Figure~\ref{fig:recursion-tree} for an example.

We define the \emph{weighted depth} of a node $x$ to the sum of the labels of
the edges on the path from the root to $x$.
For an internal node $x$ in $\tree{A}{G}{k}$, define the
\emph{branching vector} of $x$, denoted $\bv{x}$, to be a vector containing
the labels of the edges between $x$ and its children.
Define the \emph{branching number} of a vector $(a_1,\ldots,a_t)$ to be
the largest root of $P(x) = 1-\sum_{i=1}^t x^{-a_i}$.
We define the branching number of a node $x$ in $\tree{A}{G}{k}$ to be the
branching number of $\bv{x}$.
The running time of the algorithm $A$ can be bounded by bounding the number of
leaves in $\tree{A}{G}{k}$.
The number of leaves in $\tree{A}{G}{k}$ is $O(c^k)$, where $c$ is the maximum
branching number of a node in the tree.

An approach for obtaining a better bound on the number of leaves in the
recursion tree is to treat several steps of the algorithm as one step.
This can be viewed as modifying the tree $\tree{A}{G}{k}$ by contracting
some edges.
If $x$ is a nodes in $\tree{A}{G}{k}$ and $y$ is a child of $x$,
\emph{contracting} the edge $(x,y)$ means deleting the node $y$ 
and replacing every edge $(y,z)$ between $y$ and a child $z$ of $y$
with an edge $(x,z)$.
The label of $(x,z)$ is equal to the label of $(x,y)$ plus the label of $(y,z)$.
See Figure~\ref{fig:contract}.

\begin{figure}
\centering
\subfigure[$\Hv$\label{fig:Hv}]{\includegraphics{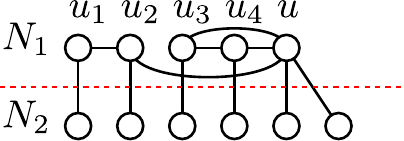}}
\quad
\subfigure[$\vctree{\Hv}{k}$]{\includegraphics{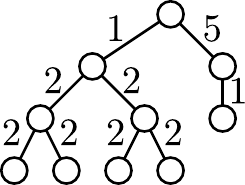}}
\quad
\subfigure[Contraction\label{fig:contract}]{\quad\includegraphics{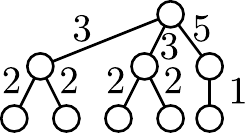}}
\quad
\subfigure[$\vctreeb{4}{\Hv}{k}$\label{fig:T4}]{\includegraphics{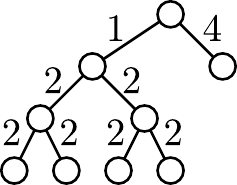}}
\caption{Example of a recursion tree of algorithm $\vcalgname$ from
Section~\ref{sec:dominating-alg}.
Figure~(a) shows a graph $\Hv$.
When applying algorithm $\vcalgname$ on $\Hv$ and $k \geq 6$,
the algorithm applies Rule~VC.\ref{vcrule:degree-3} on $u$.
In the branch $\Hv-u$, the algorithm applies Rule~VC.\ref{vcrule:N1-path} on
$u_1,u_2$, and in the two resulting branches the algorithm applies
Rule~VC.\ref{vcrule:N1-path} on $u_3,u_4$.
In the branch $\Hv-N_{\Hv}(u)$, the algorithm applies
Rule~VC.\ref{vcrule:degree-1} on $u_1$.
The recursion tree $\vctree{\Hv}{k}$ is shown in Figure~(b).
Figure~(c) shows the tree $\vctree{\Hv}{k}$ after contracting the edge between
the root and its left child.
The top recursion tree $\vctreeb{4}{\Hv}{k}$ is shown in
Figure~(d).\label{fig:recursion-tree}}
\end{figure}

\subsection{Analysis of the algorithm of Section~\ref{sec:dominating-alg}}

In order to analyze the algorithm of Section~\ref{sec:cvd-alg}, we want to
enumerate all possible recursion trees for the algorithm.
However, since the number of recursion trees is unbounded,
we will only consider a small part of the recursion tree,
called \emph{top recursion tree}.
Suppose that we know that $\vc{\Hv} \geq \alpha$ for some integer $\alpha$.
Then, mark every node $x$ in $\vctree{\Hv}{k}$ with weighted depth less than
$\alpha$.
Additionally, if $x$ is a node with weighted depth $d < \alpha$ whose
branching vector is $(1,2)$
then mark all the descendants of $x$ with distance at most $\alpha-d-1$ from
$x$.
Now define the top recursion tree $\vctreeb{\alpha}{\Hv}{k}$ to be the subtree
of $\vctree{\Hv}{k}$ induced by the marked vertices and their children.
The labels of the edges of $\vctreeb{\alpha}{\Hv}{k}$ are modified as follows.
If a node $x$ has a single child $y$ and the label of $(x,y)$ is $a$ for
$a > 1$, change the label of the edge to $1$.
If $x$ has two children $x_1,x_2$, let $a_1,a_2$ be the labels of the edges
$(x,x_1),(x,x_2)$, respectively.
If $a_1 = 1$ and $a_2 > 4$, replace the label of $(x,x_2)$ with $4$.
If $a_1 \geq 2$ and $a_2 > 3$, replace the labels of $(x,x_1)$ and $(x,x_2)$
with $2$ and $3$, respectively.
The reason for changing the labels of edges in the top recursion tree is that
this reduces the number of possible top recursion trees.

We now show some properties of the tree $\vctreeb{\alpha}{\Hv}{k}$ when
$v$ is a vertex with maximum degree in $G$.

We define an ordering $\prec$ on the brancing vectors of the nodes of a top
recursion tree.
Define $(1) \prec (1,4) \prec (1,3) \prec (2,2) \prec (2,3) \prec (1,2)$.
This order corresponds to the order of the reduction and branching rules that
generate these vectors.
That is, a node $x$ in a top recursion tree has branching vector $(1)$
if the rule that algorithm $\vcalgname$ applied in the corresponding recursive
call is either VC.\ref{vcrule:degree-1},
VC.\ref{vcrule:N1-cycle}, or VC.\ref{vcrule:alternating-cycle}.
If the branching vector is $(1,4)$ or $(1,3)$ then the algorithm applied
Rule~VC.\ref{vcrule:degree-3}.
If the branching vector is $(2,2)$ the algorithm applied
Rule~VC.\ref{vcrule:N1-path}.
If the branching vector is $(2,3)$ the algorithm applied
Rule~VC.\ref{vcrule:N1-path} or Rule~VC.\ref{vcrule:alternating-path}.
If the branching vector is $(1,2)$ the algorithm applied
Rule~VC.\ref{vcrule:alternating-path}
(note that in this case, the corresponding graph is a skien).

For the following lemmas, suppose that $v$ is a vertex with maximum degree in
$G$, and consider a top recursion tree $T = \vctreeb{\alpha}{\Hv}{k}$ for some $\alpha$
and $k$.
\begin{lemma}\label{lem:legal-1}
If $x,y$ are nodes in $T$ such that $y$ is a child of $x$ then
$\bv{x} \prec \bv{y}$.
\end{lemma}
\begin{proof}
The lemma follows directly from the definition of algorithm $\vcalgname$ and
the definition of $T$.
\end{proof}

For the next two properties of $T$, we first give the following
lemma.
\begin{lemma}\label{lem:N1-neighbors}
Let $v$ be a vertex with maximum degree in $G$.
In the graph $\Hv$,
a vertex $u \in N_1$ with $s$ neighbors in $N_2$ has at least $s$ neighbors
in $N_1$.
\end{lemma}
\begin{proof}
Let $d = \deg_G(v) = |N_1|$.
Let $s'$ be the number of neighbors of $u$ in $N_1$ (in the graph $\Hv$).
By the definition of $\Hv$, $\deg_G(u) = s+(d-1-s')+1$
(in $G$, $u$ has $s$ neighbors in $N_2$ and $d-1-s'$ neighbors in $N_1$).
Since $\deg_G(u) \leq \deg_G(v) = d$, the lemma follows.
\end{proof}

\begin{lemma}\label{lem:legal-2}
If there is a node $x$ in $T$ with branching vector 
then the branching vector of the root of $T$ is either $(1)$ or $(1,4)$.
\end{lemma}
\begin{proof}
Suppose that the node $x$ corresponds to the recursive call
$\vcalg{H}{k'}$.
By definition, $H$ is a skein, so there is a vertex $u \in N_1$ such that
$u$ has two neighbors in $N_2$ in the graph $H$.
Since $H$ is a subgraph of $\Hv$, $u$ also has two neighbors in $N_2$ in the
graph $\Hv$.
By Lemma~\ref{lem:N1-neighbors}, $\deg_{\Hv}(u) \geq 4$, and the lemma
follows from the definition of the algorithm and the definition of $T$.
\end{proof}

\begin{lemma}\label{lem:legal-3}
If the branching vector of the root of $T$ is $(1,3)$ or $(1,4)$
then the branching vector of the left child of the root is not $(1,2)$.
\end{lemma}
\begin{proof}
Suppose conversely that the label of the left child of the root is $(1,2)$.
By definition, when algorithm $\vcalgname$ is called on $\Hv$,
it applies Rule~VC.\ref{vcrule:degree-3},
and let $u$ be the vertex of $\Hv$ on which the rule is applied.
From the assumption that the branching vector of the left child of the root
is $(1,2)$, $\Hv-u$ is a skien.
Denote by $K$ the vertices of the seagulls in $\Hv-u$.
Every neighbor of $u$ (in $\Hv$) which is not in $K$ has degree~1 (in $\Hv$).
Since Rule~VC.\ref{vcrule:degree-1} was not applied on $\Hv$,
every neighbor of $u$ (in $\Hv$) which is not in $K$ is in $N_2$.
Therefore, $N_1$ consists of the centers of the seagulls of $\Hv-u$,
and possibly $u$ and isolated vertices.
Let $w$ be some center of a seagull in $\Hv-u$.
In $\Hv$, $w$ have at most one neighbor in $N_1$ and at least two neighbors in
$N_2$, contradicting Lemma~\ref{lem:N1-neighbors}.
We obtain that the label of the left child of the root is not $(1,2)$.
\end{proof}

\subsection{Analysis of the main algorithm}
\newcommand{\cvdtree}[2]{\tree{\cvdalgname}{#1}{#2}}
\newcommand{\combinedtree}[2]{\tree{}{#1}{#2}}
\newcommand{\combinedtreeb}[2]{\tree{\mathrm{C}}{#1}{#2}}

We now analyze the main algorithm. Our method is based on the analysis
in~\cite{boral2016fast} with some changes.

To simplify the analysis, suppose that Rule~VC.\ref{vcrule:terminate-1}
returns a list containing an arbitrary vertex cover of $H$
(e.g.\ the set $N_1$).
This change increases the time complexity of the algorithm.
Thus, it is suffices to bound the time complexity of the modified algorithm.

To analyze the algorithm, we define a tree $\combinedtree{G}{k}$ that
represents the recursive calls to both $\cvdalgname$ and $\vcalgname$.
Consider a node $x$ in $\cvdtree{G}{k}$, corresponding to a recursive call
$\cvdalg{G'}{k'}$.
Suppose that in the recursive call $\cvdalg{G'}{k'}$, the algorithm applies
Rule~B\ref{brule:vc3}.
Recall that in this case, the algorithm branches on $\twins{v}$
and on every set in $\F{v}{k'}$. Denote $\F{v}{k'} = \{X_1,\ldots,X_t\}$.
In the tree $\cvdtree{G}{k}$, $x$ has $t+1$ children $y,x_1,\ldots,x_t$.
The label of the edge $(x,y)$ is $|\twins{v}|$,
and the label of an edge $(x,x_i)$ is $|X_i|$.
The tree $\combinedtree{G}{k}$ also contains the nodes $x,y,x_1,\ldots,x_t$.
In $\combinedtree{G}{k}$, $x$ has two children $y$ and $x'$.
The labels of the edges $(x,y)$ and $(x,x')$ are $|\twins{v}|$ and $0$,
respectively.
The node $x'$ is the root of the tree $\vctree{\Hv}{k'}$.
The nodes $x_1,\ldots,x_t$ are the leaves of $\vctree{\Hv}{k'}$.
See Figure~\ref{fig:B3-tree} for an example.

\begin{figure}
\centering
\subfigure[]{\includegraphics{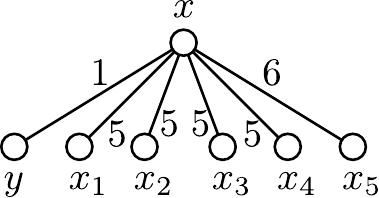}}
\quad
\subfigure[]{\includegraphics{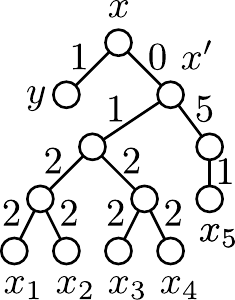}}
\quad
\subfigure[]{\includegraphics{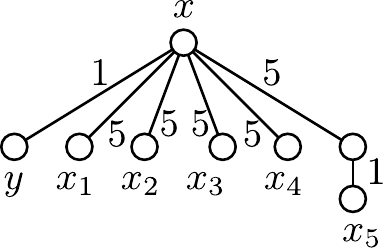}}
\caption{Example of the definition of $\combinedtreeb{G}{k}$.
Figure~(a) shows the node $x$ and its children in $\cvdtree{G}{k}$.
Suppose that the graph $\Hv$ in the corresponding call to $\cvdalgname$ is the
graph in Figure~\ref{fig:Hv}.
In this case, $\F{v}{k'}$ consists of four sets of size~5 and one set of size~6.
Since $\vc{\Hv} = 5$, the algorithm applies Rule~B\ref{brule:vc3} and
branches on $\{v\}$ and on every set in $\F{v}{k'}$.
Therefore, in the tree $\cvdtree{G}{k}$, $x$ has 6 children $y,x_1,\ldots,x_5$
and the labels of the edges between $x$ and its children are $1,5,5,5,5,6$.
Figures~(b) and~(c) show the corresponding subtree in $\combinedtree{G}{k}$
and $\combinedtreeb{G}{k}$, respectively.
Note that the edges that were contracted are $(x,x')$ and all the edges
of $\vctree{\Hv}{k'}$ except the edge between $x_5$ and its parent.
The branching vector of $x$ in $\combinedtreeb{G}{k}$ is $(1,5,5,5,5,5)$.
Note that this branching vector is greater than $(1,5,5,5,5,4)$, which is the
branching vector obtained by adding an element 1 to $(5,5,5,5,4)$, where
the latter vector is the sequence of weighted depths of the leaves in
the top recursion tree $\vctreeb{4}{\Hv}{k}$ in
Figure~\ref{fig:T4}.\label{fig:B3-tree}}
\end{figure}

Similarly, if in the recursive call $\cvdalg{G'}{k'}$ that corresponds to $x$
the algorithm applies Rule~B\ref{brule:vc2}, then the algorithm branches on
every set in $\F{v}{k'}$ and on $\{v\} \cup S$ for every $S \in \F{w}{k'-1}$.
Denote $\F{v}{k'} = \{X_1,\ldots,X_t\}$ and
$\F{w}{k'-1} = \{Y_1,\ldots,Y_s\}$.
In the tree $\cvdtree{G}{k}$, $x$ has $s+t$ children
$y_1,\ldots,y_s,x_1,\ldots,x_t$.
The label of an edge $(x,x_i)$ is $|X_i|$, and the label of an edge $(x,y_i)$
is $1+|Y_i|$.
In the tree $\combinedtree{G}{k}$, $x$ has two children $y'$ and $x'$.
The labels of the edges $(x,y')$ and $(x,x')$ are $1$ and $0$, respectively.
The node $x'$ is the root of the tree $\vctree{\Hv}{k'}$, and
$x_1,\ldots,x_t$ are the leaves of this tree.
The node $y'$ is the root of the tree $\vctree{\Hw}{k'-1}$,
and $y_1,\ldots,y_s$ are the leaves of this tree.

Finally,
suppose that in the recursive call $\cvdalg{G'}{k'}$ that corresponds to $x$,
the algorithm applies Rule~B\ref{brule:vc1}.
In this case, the algorithm branches on every set in $\Fv$.
In the tree $\cvdtree{G}{k}$, $x$ has $t = |\F{v}{k'}|$ children
$x_1,\ldots,x_t$.
In $\combinedtree{G}{k}$, the node $x$ is the root of the tree
$\vctree{\Hv}{k'}$, and $x_1,\ldots,x_t$ are the leaves of this tree.


Our goal is to analyze the number of leaves in $\combinedtree{G}{k}$.
For this purpose, we perform edge contractions on $\combinedtree{G}{k}$ to
obtain a tree $\combinedtreeb{G}{k}$.
Consider a node $x$ in $\cvdtree{G}{k}$ that corresponds to a recursive call
$\cvdalg{G'}{k'}$.
Suppose that in the recursive call $\cvdalg{G'}{k'}$, the algorithm applies
Rule~B\ref{brule:vc3}.
Using the same notations as in the paragraphs above, we contract the following
edges:
\begin{inparaenum}[(1)]
\item
The edge $(x,x')$.
\item
The edges of $\vctree{\Hv}{k'}$ that are present in
\end{inparaenum}
$\vctreeb{\alpha}{\Hv}{k'}$, where $\alpha = 3$ if $\vc{\Hv} = 3$
and $\alpha = 4$ if $\vc{\Hv} \geq 4$.
We note that the distinction between the cases $\vc{\Hv} = 3$ and
$\vc{\Hv} \geq 4$ will be used later in the analysis.
If in the recursive call $\cvdalg{G'}{k'}$ the algorithm applies
Rule~B\ref{brule:vc2}, we contract the following edges.
\begin{inparaenum}[(1)]
\item
The edges $(y,y')$ and $(x,x')$.
\item
The edges of $\vctree{\Hv}{k'}$ that are present in $\vctreeb{2}{\Hv}{k'}$.
\end{inparaenum}
Note that we do not contract the edges of $\vctreeb{2}{\Hw}{k'-1}$.
If in the recursive call $\cvdalg{G'}{k'}$ the algorithm applies
Rule~B\ref{brule:vc1}, we do not contract edges.
Therefore, the branching vector of $x$ in this case is at least $(1)$ or
at least $(1,2)$.

The nodes in $\combinedtreeb{G}{k}$ that correspond to nodes in
$\cvdtree{G}{k}$ are called \emph{primary nodes}, and the remaining nodes
are \emph{secondary nodes}.
Note that secondary nodes with two children have branching vectors that are at
least $(1,2)$, and therefore their branching numbers are at most 1.619.
The branching numbers of the primary nodes are estimated as follows.
Let $x$ be a primary node and suppose that in the corresponding recursive call
$\cvdalg{G'}{k'}$, the algorithm applies Rule~B\ref{brule:vc3}.
Using the same notations as in the paragraphs above,
the branching vector of $x$ is at least $(1,c_1,\ldots,c_t)$,
where $(c_1,\ldots,c_t)$ are the weighted depths of the leaves of the top
recursion tree $\vctreeb{\alpha}{\Hv}{k'}$.
See Figure~\ref{fig:B3-tree}.
If in the recursive call $\cvdalg{G'}{k'}$, the algorithm applies
Rule~B\ref{brule:vc2}, then the branching vector of $x$ is at least
$(2,c_1,\ldots,c_t)$ or at least $(2,3,c_1,\ldots,c_t)$,
where $(c_1,\ldots,c_t)$ are the weighted depths of the leaves of the top
recursion tree $\vctreeb{2}{\Hv}{k'}$
(this follows from the fact that in $\combinedtree{G}{k}$ the node $y'$ has
branching vector of at least $(1)$ or at least $(1,2)$).

From the discussion above, we can bound the branching numbers of the nodes
of $\combinedtreeb{G}{k}$ as follows.
Generate all possible top recursion trees $\vctreeb{2}{\Hv}{k'}$.
For each tree, compute the branching number of $(2,3,c_1,\ldots,c_t)$,
where $(c_1,\ldots,c_t)$ are the weighted depths of the leaves of the tree.
Additionally, generate all possible top recursion trees
$\vctreeb{\alpha}{\Hv}{k'}$ for $\alpha \in \{3,4\}$.
For each tree, compute the branching number of $(1,c_1,\ldots,c_t)$,
where $(c_1,\ldots,c_t)$ are the weighted depths of the leaves of the tree.
The maximum branching number computed is an upper bound on the branching
numbers of the nodes of $\combinedtreeb{G}{k}$.
Since the number of possible top recursion trees is relatively large,
we used a Python script to generate these trees.
The script uses Lemmas~\ref{lem:legal-1}, \ref{lem:legal-2},
and~\ref{lem:legal-3} to reduce the number of generated trees.
The five branching vectors with largest branching numbers generated by the
script are given in Table~\ref{tab:top5}.
For the rest of the section, we consider the first four cases in
Table~\ref{tab:top5}. For each case we either
show that the case cannot occur, or give a better branching vector for the case.
Therefore, the largest branching number of a node in $\combinedtreeb{G}{k}$
is at most 1.811, and therefore the time complexity of the algorithm is
$O^*(1.811^k)$.

\begin{table}
\centering
\begin{tabular}{llcll}
\toprule
Case & $\vc{\Hv}$ & Top recursion tree & branching vector & branching number \\
\midrule
1 & 2 & \includegraphics{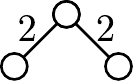} & $(2,3,2,2)$   & 1.880 \\
2 & 3 & \includegraphics{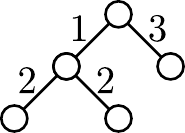} & $(1,3,3,3)$   & 1.864 \\
3 & 2 & \includegraphics{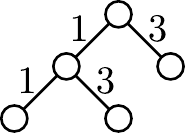} & $(2,3,2,4,3)$ & 1.840 \\
4 & 3 & \includegraphics{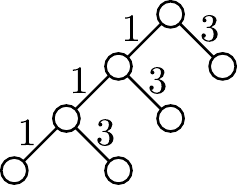} & $(1,3,5,4,3)$ & 1.840 \\
5 & 3 & \includegraphics{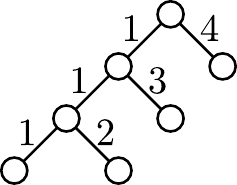} & $(1,3,4,4,4)$ & 1.811 \\
\bottomrule
\end{tabular}
\caption{The five branching vectors with largest branching numbers generated by
the Python script.\label{tab:top5}}
\end{table}

\paragraph{Case 1}
In this case, the algorithm applies Rule~B\ref{brule:vc2},
and when algorithm $\vcalgname$ is called on $\Hv$, the algorithm
applies Rule~VC.\ref{vcrule:N1-path} on a path $u_0,u_1,u_2,u_3$,
where $u_1,u_2 \in N_1$ and $u_0,u_3 \in N_2$.
Since $\vc{\Hv} = 2$ and $v$ does not have twins (in other words, $\Hv$ does not
have isolated vertices),
the vertices of $\Hv$ are $u_0,u_1,u_2,u_3$.
It follows that $\deg_G(v) = 2$.
Since $v$ is vertex of maximum degree in $G$, this contradicts the assumption
that Rule~R\ref{rrule:path-or-cycle} cannot be applied.
Therefore, case~1 cannot occur.

\paragraph{Case 2}
In this case, the algorithm applies Rule~B\ref{brule:vc3},
$\vc{\Hv} = 3$,
and when algorithm $\vcalgname$ is called on $\Hv$, the algorithm
applies Rule~VC.\ref{vcrule:degree-3} on a vertex $u$ with degree~3.
In the branch $\Hv-u$, the algorithm applies Rule~VC.\ref{vcrule:N1-path} on
a path $u_0,u_1,u_2,u_3$, where $u_1,u_2 \in N_1$ and $u_0,u_3 \in N_2$.

If $\Hv$ contains isolated vertices then the branching vector is at least
$(2,3,3,3)$ and the branching number is at most 1.672.
We now assume that $\Hv$ does not contain isolated vertices.

Since Rule~VC.\ref{vcrule:N1-path} was applied on $\Hv-u$, we
obtain the following.
\begin{observation}\label{obs:Nu1}
$N_{\Hv}(u_1) \subseteq \{u_0, u_2, u\}$ and
$N_{\Hv}(u_2) \subseteq \{u_1, u_3, u\}$.
\end{observation}


Let $X$ be a vertex cover of $\Hv$ of size~3.
In order to cover the edge $(u_1,u_2)$, $X$ must contain either $u_1$ or $u_2$.
Suppose without loss of generality that $u_1 \in X$.
Since $X$ covers the edge $(u_2,u_3)$, there is an index $i\in \{2,3\}$
such that $u_i \in X$.

\begin{lemma}\label{lem:case-2}
$N_1 \subseteq \{u_1,u_2,u\}$.
\end{lemma}
\begin{proof}
Suppose conversely that there is a vertex $w \in N_1 \setminus \{u_1,u_2,u\}$.
We have that $\deg_{\Hv}(w) \geq 2$ since otherwise, $\deg_{\Hv}(w) = 1$
(recall that we assumed that there are no isolated vertices),
contradicting the fact that Rule~VC.\ref{vcrule:degree-1} was not applied
on $\Hv$.
Additionally, $\deg_{\Hv}(w) \leq \deg_{\Hv}(u) = 3$ (since $u$ is a vertex
of $\Hv$ with maximum degree).
By Lemma~\ref{lem:N1-neighbors}, $w$ has at most one neighbor in $N_2$.
Additionally, if $w$ is adjacent to $u$ then $\deg_{\Hv}(w) = 3$,
otherwise $\deg_{\Hv-u}(w) = 1$,
contradicting the fact that Rule~VC.\ref{vcrule:degree-1} was not applied
on $\Hv-u$.
Since $w$ is not adjacent to $u_1$ and $u_2$ (by Observation~\ref{obs:Nu1}),
we obtain that $w$ is adjacent to  $w' \in N_1 \setminus \{u_1,u_2,u\}$.
Since $X$ covers the edge $(w,w')$, $X$ contains either $w$ or $w'$.
Without loss of generality, suppose that $w \in X$.
Since $X$ is a vertex cover and $w' \notin X$, we have that
$N_{\Hv}(w') \subseteq X$.
Using the same arguments as above, we have that $\deg_{\Hv}(w') \geq 2$.
By Observation~\ref{obs:Nu1}, $w$ is not adjacent to $u_1$ and $u_2$.
Therefore, $i = 3$ and $w'$ is adjacent to $u_3$.
Since $X$ is a vertex cover and $u \notin X$, we have that
$N_{\Hv}(u) \subseteq X$. From the fact that $\deg_{\Hv}(u) = 3$ we obtain that
$N_{\Hv}(u) = X$, and in particular, $u$ is adjacent to $u_3$,
which implies that $u \in N_1$.
Now, $\deg_{\Hv}(u_3) = 3 = \deg_{\Hv}(u)$
(the neighbors of $u_3$ are $u_2,u,w'$), $u_3 \in N_2$ and $u \in N_1$.
We obtain a contradiction to the choice of $u$ when
Rule~VC.\ref{vcrule:degree-3} is applied on $\Hv$.
Therefore, the lemma is correct.
\end{proof}

From Lemma~\ref{lem:case-2} we obtain that $u \in N_1$:
Suppose conversely that $u \in N_2$.
Therefore, all the neighbors of $u$ are in $N_1$.
However, $N_1 = \{u_1, u_2\}$,
contradicting the fact that $\deg_{\Hv}(u) = 3$.
Therefore, $u \in N_1$.

We have that $u$ has at least two neighbors in $N_1$, otherwise
$u$ have at most one neighbor in $N_1$ and at least two neighbors in $N_2$,
contradicting Lemma~\ref{lem:N1-neighbors}.
Since $N_1 = \{u,u_1,u_2\}$, we obtain that $N_1$ is a clique.
Additionally, every vertex in $N_1$ has exactly one neighbor in $N_2$
(note that these neighbors are not necessarily distinct).

Now, consider the application of Rule~B\ref{brule:vc3} on $G$.
In the branch $G-v$, the vertices $u,u_1,u_2$ have degree~1.
Therefore, the algorithm applies either a reduction rule on $G-v$ or
Rule~B\ref{brule:vc1}.
Note that Rule~R\ref{rrule:clique} cannot be applied: Conversely, if there is
a connected component $C$ in $G-v$ which is a clique then $C$ must contain a
vertex $w \in \{u,u_1,u_2\}$.
However, $w$ has a single neighbor $w'$, and
$\deg_{G-v}(w') = \deg_G(w') \geq 2$
(since Rule~B\ref{brule:vc1} was not applied on $G$).
Therefore, $C$ is not a clique, a contradiction.

If the algorithm applies Rule~B\ref{brule:vc1}, then the algorithm
for computing $\Fv$ applies either a reduction rule on $\Hv$ or a branching
rule with branching vector at least $(1,2)$.
Therefore, the branching vector for Case~2 is at least
$(1+1,3,3,3) = (2,3,3,3)$ or at least $(1+(1,2),3,3,3) = (2,3,3,3,3)$.
The branching number is at most 1.797.

\paragraph{Case 3}
In this case, the algorithm applies Rule~B\ref{brule:vc2},
and when algorithm $\vcalgname$ is called on $\Hv$, the algorithm
applies Rule~VC.\ref{vcrule:degree-3} on a vertex $u$ with degree~3.
In the branch $\Hv-u$, the algorithm applies Rule~VC.\ref{vcrule:degree-3} on
a vertex $u'$ with degree~3.

Since $\vc{\Hv} = 2$ and $\deg_{\Hv}(u) = \deg_{\Hv}(u') = 3$, we have
that the unique vertex cover of $\Hv$ of size~2 is $\{u,u'\}$.
Therefore, the graph $\Hv-u$ consists of a 3-star whose center is $u'$.
Since Rule~VC.\ref{vcrule:degree-1} was not applied on $\Hv-u$, the leaves of
the star are in $N_2$ and therefore the center $u'$ of the star is in $N_1$.
We now have that $u'$ has at most one neighbor in $N_1$ and at least 3 neighbors
in $N_2$, contradicting Lemma~\ref{lem:N1-neighbors}.
Therefore, case~3 cannot occur.

\paragraph{Case 4}
In this case, the algorithm applies Rule~B\ref{brule:vc3},
$\vc{\Hv} = 3$,
and when algorithm $\vcalgname$ is called on $\Hv$, the algorithm
applies Rule~VC.\ref{vcrule:degree-3} on a vertex $u$ with degree~3.
In the branch $\Hv-u$, the algorithm applies Rule~VC.\ref{vcrule:degree-3} on
a vertex $u'$ with degree~3, and in the branch $\Hv-\setminus \{u,u'\}$ the
algorithm applies Rule~VC.\ref{vcrule:degree-3} on a vertex $u''$ with
degree~3.

Note that $u,u'$ are not adjacent otherwise
$\deg_{\Hv}(u') = \deg_{\Hv-u}(u')+1 = 4$, contradicting the choice of $u$ when
Rule~VC.\ref{vcrule:degree-3} was applied on $\Hv$.
Using the same argument we have that $u,u',u''$ is an independent set in $\Hv$.

\begin{lemma}
$\{u,u',u''\}$ is a vertex cover of $\Hv$.
\end{lemma}
\begin{proof}
Let $X$ be a vertex cover of $\Hv$ of size~3.
Suppose that $u'' \notin X$. Therefore, $N_{\Hv}(u'') \subseteq X$.
Since $\deg_{\Hv}(u'') = \deg_{\Hv \setminus \{u,u'\}}(u'') = 3$, it follows
that $X = N_{\Hv}(u'')$.

The vertices $u$ and $u'$ are not adjacent to $u''$ and therefore
$u,u' \notin X$.
From the fact that $X$ is a vertex cover we obtain that $N_{\Hv}(u) \subseteq X$
and $N_{\Hv}(u') \subseteq X$.
Therefore, $N_{\Hv}(u) = N_{\Hv}(u') = X$.
We have shown that every vertex in $X$ is adjacent to $u,u',u''$.
Since every vertex in $\Hv$ have degree at most~3 (recall that $u$ is a vertex
of maximum degree in $\Hv$ and $\deg_{\Hv}(u)=3$), it follows that
$X$ is an independent set.
Therefore, $\{u,u',u''\}$ is also a vertex cover of $\Hv$.

Now suppose that $X$ is a vertex cover of $\Hv$ of size~3 and $u''\in X$.
We have that $u' \in X$ otherwise $N_{\Hv}(u')\subseteq X$ and therefore
$|X| \geq 4$ (as $u'' \notin N_{\Hv}(u')$).
Similarly, $u \in X$.
\end{proof}

From the fact that $\{u,u',u''\}$ is a vertex cover of $\Hv$, we have that
the graph $\Hv-\{u,u'\}$ consists of a 3-star whose center is $u''$ and
isolated vertices.
Since Rule~VC.\ref{vcrule:degree-1} was not applied on $\Hv-\{u,u'\}$,
the leaves of the star are in $N_2$ and therefore the center $u''$ of the star
is in $N_1$.
We now have that $u''$ does not have neighbors in $N_1$ and 3 neighbors
in $N_2$, contradicting Lemma~\ref{lem:N1-neighbors}.
Therefore, case~4 cannot occur.

\bibliographystyle{plain}
\bibliography{cluster,parameterized}

\end{document}